\documentclass[a4paper,USenglish]{lipics-v2016}
\usepackage[utf8]{inputenc}
 
\usepackage{microtype}

\usepackage{pgfplotstable}
\usepackage{pgfplots}
\usetikzlibrary{shapes}
\usepackage[utf8]{inputenc}
\RequirePackage{amsmath}
\DeclareMathOperator*{\opt}{opt}
\usetikzlibrary{shapes}
\usepackage{pifont}
\newcommand{\cmark}{\ding{51}}%
\newcommand{\xmark}{\ding{55}}%
\pgfplotsset{tikzDefaults/.style={
		enlargelimits=0,
		ticklabel style = {font=\scriptsize},
		label style = {font=\scriptsize},
		legend style = {font=\scriptsize},
		thin,
		mark size=1pt,
	}}
	
\pgfplotsset{compat=newest}

\usepackage{beramono}
\usepackage{listings}
\usepackage{xcolor}

\lstdefinelanguage{Julia}%
{morekeywords={abstract,break,case,catch,const,continue,do,else,elseif,%
		end,export,false,for,function,immutable,import,importall,if,in,%
		macro,module,otherwise,quote,return,switch,true,try,type,typealias,%
		using,while},%
	sensitive=true,%
	alsoother={\$},%
	morecomment=[l]\#,%
	morecomment=[n]{\#=}{=\#},%
	morestring=[s]{"}{"},%
	morestring=[m]{'}{'},%
}[keywords,comments,strings]%

\title{Kleinberg's Grid Reloaded}

\author[1]{Fabien Mathieu}
\affil[1]{Nokia Bell Labs, route de Villejust, 91620 Nozay\\
  \texttt{fabien.mathieu@nokia-bell-labs.com}}
\authorrunning{Fabien Mathieu} 

\Copyright{Fabien Mathieu}

\subjclass{C.2.2 Routing Protocols; C.2.4 Distributed Systems}
\keywords{Small-World Routing; Kleinberg's Grid; Simulation; Rejection Sampling}

\EventEditors{Panagiota Fatourou, Ernesto Jiménez and Fernando Pedone}
\EventNoEds{3}
\EventLongTitle{20th International Conference on Principles of Distributed Systems (OPODIS 2016)}
\EventShortTitle{OPODIS 2016}
\EventAcronym{OPODIS}
\EventYear{2016}
\EventDate{December 13--16, 2016}
\EventLocation{Madrid, Spain}
\EventLogo{}
\SeriesVolume{20}
\ArticleNo{21}

\begin{document}

\maketitle

\begin{abstract}
One of the key features of small-worlds is the ability to route messages with few hops only using local knowledge of the topology. In 2000, Kleinberg proposed a model based on an augmented grid that asymptotically exhibits such property.

In this paper, we propose to revisit the original model from a simulation-based perspective. Our approach is fueled by a new algorithm that uses dynamic rejection sampling to  draw augmenting links. The speed gain offered by the algorithm enables a detailed numerical evaluation. We show for example that in practice, the augmented scheme proposed by Kleinberg is more robust than predicted by the asymptotic behavior, even for very large finite grids. We also propose tighter bounds on the performance of Kleinberg's routing algorithm. At last, we show that fed with realistic parameters, the model gives results in line with real-life experiments.
\end{abstract}

\section{Introduction}

In a prescient 1929 novel called \emph{Láncszemek} (in English, \emph{Chains}), 
Karinthy imagines that any two people can be connected by a small chain of personal links, using no more than \emph{five} intermediaries~\cite{karinthy29chains}.

Years later, Milgram validates the concept by conducting real-life experiments. He asks volunteers to transmit a letter to an acquaintance with the objective to reach a target destination across the United States~\cite{milgram67small,milgram69experimental}. While not all messages arrive, successful attempts reach destination after six hops in average, popularizing the notion of \emph{six degrees of separation}.

Yet, for a long time, no theoretical model could explain why and how this kind of \emph{small-world} routing works.
One of the first and most famous attempts to provide such a model is due to Jon Kleinberg \cite{Kleinberg00thesmall-world}. He proposes to abstract the social network by a grid augmented with \emph{shortcuts}. If the shortcuts follow a heavy tail distribution with a specific exponent, then a simple greedy routing can reach any destination in a short time ($ O(\log^2(n)) $ hops). On the other hand, if the exponent is wrong, then the time to reach destination becomes $ \Omega(n^\alpha) $ for some $ \alpha $. This seminal work has led to multiple studies from both the theoretical and empirical social systems communities. 

\subsection*{Contribution}

In this paper, we propose a new way to numerically benchmark the greedy routing algorithm in the original model introduced by Kleinberg. Our approach uses dynamic rejection sampling, which gives a substantial speed improvement compared to previous attempts, without making any concession about the assumptions made in the original model, which is kept untouched.

Fueled by the capacity to obtain quick and accurate results even for very large grids, we give a fresh look on Kleinberg's grid, through three independent small studies. First, we show that the model is in practice more robust than expected: for grids of given size there is quite a large range of exponents that grant short routing paths. Then we observe that the lower bounds proposed by Kleinberg in \cite{Kleinberg00thesmall-world} are not tight and suggest new bounds. Finally, we compare Kleinberg's grid to Milgram's experiment, and observe that when the grid parameters are correctly tuned, the performance of greedy routing is consistent with the \emph{six degrees of separation} phenomenon.

\subsection*{Roadmap}

Section \ref{sec:related-work} presents the original augmented grid model introduced by Kleinberg and the greedy routing algorithm. A brief overview of the main existing theoretical and experimental studies is provided, with a strong emphasis on the techniques that can be used for the numerical evaluation of Kleinberg's model.

In Section \ref{sec:simulation-design}, we give our algorithm for estimating the performance of greedy routing. We explain the principle of dynamic rejection sampling and detail why it allows to perfectly emulate Kleinberg's grid with the same speed that can be achieved by toroidal approximations. We also give a performance evaluation of the simulator based on our solution. For readers interested in looking under the hood, a fully working code (written in Julia) is given in Appendix \ref{sec:code}.

To show the algorithm benefits, we propose in Section \ref{sec:applications} three small studies that investigate Kleinberg's model from three distinct perspectives: robustness of greedy routing with respect to the shortcut distribution (Section \ref{sec:efficient-enough-exponents});  tightness of the existing theoretical bounds (Section \ref{sec:asymptotic-behavior}); emulation of Milgram's experiment within Kleinberg's model (Section \ref{sec:six-degrees-of-separation}).

\section{Model and Related Work}\label{sec:related-work}

We present here the model and notation introduced by Kleinberg in~\cite{Kleinberg00thesmall-world,Kleinberg00navigation}, some key results, and a brief overview of the subsequent work on the matter.

\subsection{Kleinberg's Grid}
\label{sec:rw_kb}

In \cite{Kleinberg00thesmall-world}, Kleinberg considers a model of directed random graph $ G(n,r,p,q) $, where $ n, p, q $ are positive integers and $ r $ is a non-negative real number. A graph instance is built from a square lattice of $ n\times n $ nodes with Manhattan distance $ d $: if $ u = (i,j) $ and $ v = (k, l) $, then $ d(u,v) = |i-j|+|k-l| $.
$ d $ represents some natural proximity (geographic, social, \ldots) between nodes. Each node has some \emph{local} neighbors and $ q $ \emph{long range} neighbors. The local neighbors of a node $ u $ are the nodes $ v $ such that $ d(u,v)\leq p $. The $ q $ long range neighbors of $ u $, also called \emph{shortcuts}, are drawn independently and identically as follows: the probability that a given long edge starting from $ u $ arrives in $ v $ is proportional to $ (d(u,v))^{-r} $.

The problem of decentralized routing in a $ G(n,r,p,q) $ instance consists in delivering a message from node $ u $ to node $ v $ in a hop-by-hop basis. At each step, the message bearer needs to choose the next hop among its neighbors. The decision can only use the lattice coordinates of the neighbors and destination. The main example of decentralized algorithm is the \emph{greedy routing}, where at each step, the current node chooses the neighbor that is closest to destination based on $ d $ (in case of ties, an arbitrary breaking rule is used).

The main metric to analyze the performance of a decentralized algorithm is the \emph{expected delivery time}, which is the expected number of hops to transmit a message between two nodes chosen uniformly at random in the graph.

This paper focuses on studying the performance of the greedy algorithm. Unless stated otherwise, we assume $ p = q = 1 $ (each node has up to four local neighbors and one shortcut). Let $ e_r(n) $ be the expected delivery time of the greedy algorithm in $ G(n,r,1,1) $.

\subsection{Theoretical Results}
\label{sec:rw_theory}

The main theoretical results for the genuine model are provided in the original papers~\cite{Kleinberg00navigation,Kleinberg00thesmall-world}, where Kleinberg proves the following:
\begin{itemize}
	\item $ e_2(n) = O(\log^2(n)) $;
	\item for $ 0\leq r < 2 $, the expected delivery time of any decentralized algorithm is $ \Omega(n^{(2-r)/3}) $;
	\item for $ r > 2 $, the expected delivery time of any decentralized algorithm is $ \Omega(n^{(r-2)/(r-1)}) $.
\end{itemize}
Kleinberg's results are often interpreted as follows: short paths are easy to find only in the case $ r = 2 $. The fact that only one value of $ r $ asymptotically works is sometimes seen as the sign that Kleinberg's model is not robust enough to explain the small-world routing proposed by Karinthy and experimented by Milgram. However, as briefly discussed by Kleinberg in \cite{kleinberg10networks}, there is in fact some margin if one considers a grid of given $ n $. This tolerance will be investigated in more details in Section \ref{sec:efficient-enough-exponents}.

While we focus here on the original model, let us give a brief, non-exhaustive, overview of the subsequent extensions that have been proposed since. Most proposals refine the model by considering other graph models or other decentralized routing algorithms. 
New graph models are for example variants of the original model (studying grid dimension or the number of shortcuts per node \cite{fraigniaud01efficient,kleinberg10networks,fraigniaud14greedy}), graphs inspired by peer-to-peer overlay networks \cite{mankun04eighbor}, or arbitrary graphs augmented with shortcuts~\cite{fraigniaud09universal,stoc2010_searchability}.
Other proposals of routing algorithms usually try to enhance the performance of the greedy one by granting the current node additional knowledge of the topology \cite{fraigniaud2006eclectism,mankun04eighbor,martel04analyzing}.

A large part of the work above aims at improving the $ O(\log^2(n)) $ bound of the greedy routing. For example, in the small-world percolation model, a variant of Kleinberg's grid with $ O(\log(n)) $ shortcuts per node, greedy routing performs in $ O(\log(n)) $ \cite{mankun04eighbor}.

\subsection{Experimental Studies}
\label{sec:rw_xp}

Many empirical studies have been made to study how routing works in real-life social networks and the possible relation with Kleinberg's model (see for example \cite{Liben-Nowell16082005,kleinberg10networks} and the references within). On the other hand, numerical evaluations of the theoretical models are more limited to the best of our knowledge. Such evaluations are usually performed by averaging $ R $ runs of the routing algorithm considered.

In \cite{Kleinberg00navigation}, Kleinberg computes $ e_r(n) $ for $ n = 20,000 $ and $ r \in [0, 2.5] $, using 1,000 runs per estimate. However, he uses a torus instead of a regular grid (this will be discussed later in the paper). In \cite{mankun04eighbor}, networks of size up to $ 2^{24} $ (corresponding to $ n = 2^{12} $ in the grid) are investigated using 150 runs per estimate.

Closer to our work, Athanassopoulos \emph{et al.} propose a study centered on numerical evaluation that looks on Kleinberg's model and some variants \cite{athan10high}. For the former, they differ from the original model by having fixed source and destination nodes. They compute $ e_r(n) $ for values of $ n $ up to 3,000 and $ r\in \{0, 1, 2, 3\} $, using 900 runs per estimate.

To compare with, in the present paper, we consider values of $ n $ up to $ 2^{24} \approx 16,000,000 $ and $ r\in [0, 3] $, with at least 10,000 runs per estimate. To explain such a gain, we first need to introduce the issue of shortcuts computation.

\subsubsection{Drawing shortcuts}

As stated in \cite{athan10high}, the main computational bottleneck for simulating Kleinberg's model comes from the shortcuts.
\begin{itemize}
	\item There are $ n^2 $ shortcuts in the grid (assuming $ q = 1 $);
	\item When one wants to a shortcut, any of the $ n^2-1 $ other nodes can be chosen with non-null probability. This can be made by inverse transform sampling, with a cost $ \Omega(n^2) $;
	\item The shortcut distribution depends on the node $ u $ considered, even if one uses relative coordinates. For example, a corner node will have $ i+1 $ neighbors at distance $ i $ for $ 1\leq i < n $, against $ 4i $ neighbors for inner nodes (as long as the ball of radius $ i $ stays inside the grid). This means that, up to symmetry, each node has a unique shortcut distribution\footnote{To take advantage of symmetry, one can consider the isometric group of a square grid, which can be built with the quarter-turn and flip operations. However, its size is 8, so even using symmetry, there are at least $ \frac{n^2}{8} $ distinct (non-isomorphic) distributions.}. This prevents from mutualising shortcuts drawings between nodes.
\end{itemize}

In the end, building shortcuts as described above for each of the $ R $ runs has a time complexity $ \Omega(Rn^4) $, which is unacceptable if one wants to evaluate $ e_r(n) $ on large grids.

The first issue is easy to address: as observed in \cite{Kleinberg00thesmall-world,mankun04eighbor,athan10high}, we can use the \emph{Principle of deferred decision}~\cite{Mitzenmacher:2005:PCR:1076315} and compute the shortcuts on-the-fly as the path is built, because they are drawn independently and a node is never used twice in a given path. This reduces the complexity to $ \Omega(Rn^2 e_r(n)) $.

\subsubsection{Torus approximation}\label{sec:torus-approximation}

To lower the complexity even more, one can approximate the grid by the torus. This is the approach adopted in \cite{Kleinberg00navigation,mankun04eighbor}. The toroidal topology brings two major features compared to a flat grid:
\begin{itemize}
	\item The distribution of the relative position of the shortcut does not depend on the originating node. This enables to draw shortcuts in advance (in bulk);
	\item There is a strong radial symmetry, allowing to draw a ``radius'' and an ``angle'' separately.
\end{itemize}

To illustrate the gain of using a torus instead of a grid, consider the drawing of $ k $ shortcuts from $ k $ distinct nodes. In a grid, if one uses inverse transform sampling for each corresponding distribution, the cost is $ \Omega(n^2 k) $. In the torus, one can compute the probabilities to be at distance $ i $ for $ i $ between $ 1 $ and $ n $ (the maximal distance in the torus), draw $ k $ radii, then choose for each drawn radius $ i $ a node uniformly chosen among those at distance $ i $. Assuming drawing a float uniformly distributed over $ [0,1) $ can be made in $ O(1) $, the main bottleneck is the drawing of radii. Using bulk inverse transform sampling, it can be performed in $ O(n+k\log(k)) $, by sorting $ k $ random floats, matching then against the cumulative distribution of radii and reverse sorting the result.

\section{Fast Estimation of Expected Delivery Time}\label{sec:simulation-design}

We now describe our approach for computing $ e_r(n) $ in the flat grid with the same complexity than for the torus approximation.

\subsection{Dynamic rejection sampling for drawing shortcuts}\label{sec:dynamic-rejection-sampling-for-drawing-shortcuts}

In order to keep low computational complexity without making any approximation of the model, we propose to draw a shortcut of a node $ u $ as follows:
\begin{enumerate}
	\item We embed the actual grid $ G $ (we use here $ G $ to refer to the lattice nodes of $ G(n,r,p,q) $) in a virtual lattice $ B_u $ made of points inside a ball of radius $ 2(n-1) $. Note that the radius chosen ensures that $ G $ is included in $ B_u $ no matter the location of $ u $.
	\item We draw a node inside $ B_u $ such that the probability to pick up a node $ v $ is proportional to $ (d(u,v))^{-r} $. This can be done in two steps (radius and angle):
	\begin{itemize}
		\item For the radius, we notice that the probability to draw a node at distance $ i $ is proportional to $ i^{1-r} $, so we pick an integer between $ 1 $ and $ 2(n-1) $ such that the probability to draw $ i $ is to $ i^{1-r}/\sum_{k=1}^{2(n-2)}k^{1-r} $.
		\item For the angle, pick an integer uniformly chosen between 1 and $ 4i $
	\end{itemize}
	\item This determines a unique point $ v $ among the $ 4i $ points at distance $ i $ from $ u $ in the virtual lattice, chosen with a probability proportional to $ (d(u,v))^{-r} $. If $ v $ belongs to the actual grid, it becomes the shortcut, otherwise we try again (back to step \#2).
\end{enumerate}

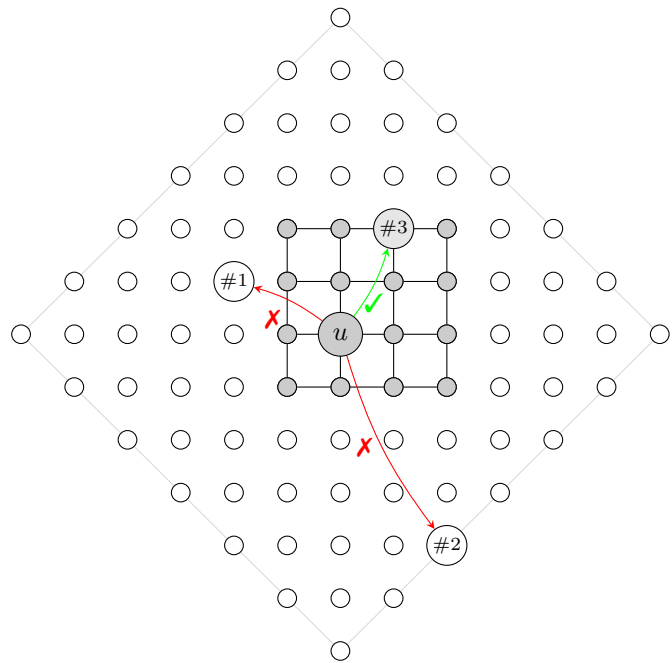
\begin{figure}
\centering	
\begin{tikzpicture}[gridstyle/.style={circle,draw,inner sep = 0, fill=gray!40,minimum size=7},
ballstyle/.style={circle, draw,inner sep = 0, minimum size = 7, fill = white},
shortstyle/.style={circle,draw, font = \scriptsize, inner sep=0pt, minimum size = 15}]
\draw[color = gray!30] (-4.2,0) -- (0,4.2) -- (4.2,0) -- (0,-4.2) -- (-4.2,0);
\foreach \r in {1,...,6}
\foreach \i in {1,...,\r}
{\node[ballstyle] (q1\r\i) at (.7*\r - .7*\i , .7*\i) {};
	\node[ballstyle] (q2\r\i) at ( - .7*\i ,.7*\r - .7*\i) {};
	\node[ballstyle] (q3\r\i) at (.7*\i - .7*\r , -.7*\i) {};
	\node[ballstyle] (q4\r\i) at (.7*\i , .7*\i - .7*\r) {};}

\foreach \x in {0,...,3}
\foreach \y in {0,...,3} 
\node [gridstyle]  (\x\y) at (.7* \x - .7,.7* \y - .7) {};

\foreach \x in {0,...,3}
\foreach \y [count=\yi] in {0,...,2}  
\draw (\x\y)--(\x\yi) (\y\x)--(\yi\x) ;

\node[circle,draw,fill=gray!40,minimum size=15] (u) at (11) {$ u $};

\node[shortstyle, fill=white] (t1) at (q232) {\#1};
\draw[->,>=stealth] (u) edge[bend angle = 10, bend right, red] node[red, below left] {\xmark} (t1);

\node[shortstyle, fill=white] (t2) at (q462) {\#2};
\draw[->,>=stealth] (u) edge[bend angle = 10, bend right, red] node[red, left] {\xmark} (t2);

\node[shortstyle, fill=gray!20] (t3) at (23) {\#3};
\draw[->,>=stealth] (u) edge[bend angle = 10, bend right, green] node[green, below] {\cmark} (t3);
\end{tikzpicture}
\caption{Main idea of the dynamic rejection sampling approach ($ n = 4 $).}
\label{fig:toyrejection}
\end{figure}

This technique, illustrated in Figure \ref{fig:toyrejection}, is inspired by the \emph{rejection sampling} method \cite{rejection}. By construction, it gives the correct distribution: the node $ v $ that it eventually returns is in the actual grid and has been drawn with a probability proportional to $ (d(u,v))^{-r} $.

We call this \emph{dynamic} rejection sampling because the sampled distribution changes with the current node $ u $. Considering $ u $ as a relative center, the actual grid $ G $ moves with $ u $ and acts like an acceptance mask.  On the other hand, the distribution over the virtual lattice $ B_u $ remains constant. This enables to draw batches of relative shortcuts that can be used over multiple runs, exactly like for the torus approximation.

The only possible drawback of this approach is the number of attempts required to draw a correct shortcut. Luckily, this number is contained.

\begin{lemma}
	\label{lem}
The probability that a node drawn in $ B_u $ belongs to $ G $ is at least $ \frac{1}{8} $.
\end{lemma}
\begin{proof}
We will prove that 
$$\frac{\sum_{v\in G\setminus \{u\}}(d(u,v))^{-r}}{\sum_{v\in B_u\setminus \{u\}}(d(u,v))^{-r}}>\frac{1}{8}\text{.}$$

We use the fact that the probability decreases with the distance combined with some geometric arguments. Let $ G_c $ be a $ n\times n $ lattice that has $ u $ as one of its corner. Let $ H_u $ the $ (2n-1)\times (2n-1) $ lattice centered in $ u $.

In terms of probability of drawing a node in $ G\setminus \{u\} $, the worst case is when $ u $ is at some corner: there is a bijection $ f $ from $ G $ to $ G_c $ such that for all $ v\in G\setminus\{u\}, d(u, f(v))\geq d(u,v) $. Such a bijection can be obtained by splitting $ G $ into $ G\cap G_c $ and three other sub-lattices that are flipped over $ G\cap G_c $ (see Figure \ref{fig:eight}). This gives 
$$\sum_{v\in G\setminus \{u\}}(d(u,v))^{-r}\geq \sum_{v\in G\setminus \{u\}}(d(u,f(v)))^{-r} = \sum_{v\in G_c\setminus \{u\}}(d(u,v))^{-r}\text{.}$$

Then we observe that the four possible lattices $ G_c $ obtained depending on the corner occupied by $ u $ fully cover $ H_u $. In fact, axis nodes are covered redundantly. This gives
$$\sum_{v\in H_u\setminus \{u\}}(d(u,v))^{-r}< 4 \sum_{v\in G_c\setminus \{u\}}(d(u,v))^{-r}\text{.}$$

Lastly, if one folds $ B_u\setminus H_u $ back into $ H_u $ like the corners of a sheet of paper, we get a strict injection $ g $ from $ B_u\setminus H_u $ to $ H_u \setminus \{u\} $ (the diagonal nodes of $ H_u $ are not covered). Moreover, for all $ v\in B_u\setminus H_u, d(u, v)\geq d(u,h(v)) $. This gives
$$\sum_{v\in B_u\setminus H_u}(d(u,v))^{-r} \leq \sum_{v\in B_u\setminus H_u}(d(u,h(v)))^{-r} < \sum_{v\in H_u\setminus \{u\}}(d(u,v))^{-r} \text{.}$$

This concludes the proof, as we get 
$$\frac{\sum_{v\in G\setminus \{u\}}(d(u,v))^{-r}}{\sum_{v\in B_u\setminus \{u\}}(d(u,v))^{-r}}\geq 
\frac{\sum_{v\in G_c\setminus \{u\}}(d(u,v))^{-r}}{\sum_{v\in H_u\setminus \{u\}}(d(u,v))^{-r}+\sum_{v\in B_u\setminus H_u}(d(u,v))^{-r}} >\frac{1}{8}\text{.}$$

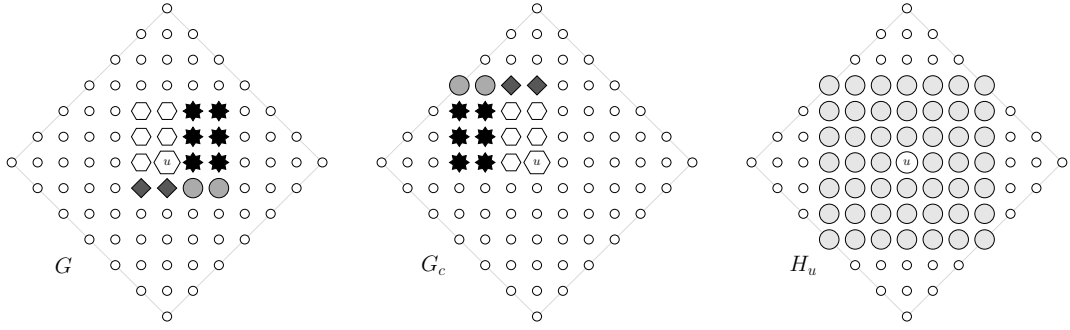
\begin{figure}
	\centering	
	\resizebox{\textwidth}{!}{%
\begin{tikzpicture}[gridstyle/.style={circle,draw,inner sep = 0, fill=gray!20,minimum size=7},
ballstyle/.style={circle, draw,inner sep = 0, minimum size = 7, fill = white},
]

\def\ball#1#2{
	\draw[color = gray!30] (#1 - 4.2, #2) -- (#1 + 0, #2 + 4.2) -- (#1 + 4.2, #2) -- (#1 + 0, #2 - 4.2) -- (#1 - 4.2, #2);
	\foreach \r in {1,...,6}
	\foreach \i in {1,...,\r}
	{
		\node[ballstyle] (q1\r\i) at (.7*\r-.7*\i+#1, .7*\i+ #2) {};
		\node[ballstyle] (q2\r\i) at (-.7*\i+#1, .7*\r-.7*\i+ #2) {};
		\node[ballstyle] (q3\r\i) at (.7*\i-.7*\r+#1, -.7*\i+ #2) {};
		\node[ballstyle] (q4\r\i) at (.7*\i+#1, .7*\i-.7*\r+ #2) {};}	
}

\def\grid#1#2#3#4#5{
	\foreach \x in {1,...,#4}
	\foreach \y in {1,...,#5}
	\node[gridstyle, #1] (\x\y) at (.7*\x-.7+#2, .7*\y-.7+#3) {};
}

\ball{0}{0}
\grid{diamond,minimum size=15, fill=black!66}{- .7}{-.7}{2}{1}
\grid{minimum size=15, fill=black!33}{.7}{-.7}{2}{1}
\grid{regular polygon,regular polygon sides=6,minimum size=15, fill=white}{- .7}{0}{2}{3}
\grid{star,star points=8,minimum size=15, fill=black}{.7}{0}{2}{3}
\node[circle,draw,fill=white!40,regular polygon,regular polygon sides=6,minimum size=15] (u) at (0,0) {$ u $};
\node (G) at (-2.8,-2.8) {\LARGE $ G $};

\ball{10}{0}
\grid{diamond,minimum size=15, fill=black!66}{10 - .7}{2.1}{2}{1}
\grid{minimum size=15, fill=black!33}{10 -2.1}{2.1}{2}{1}
\grid{regular polygon,regular polygon sides=6,minimum size=15, fill=white}{10 - .7}{0}{2}{3}
\grid{star,star points=8,minimum size=15, fill=black}{10 -2.1}{0}{2}{3}
\node[circle,draw,fill=white!40,regular polygon,regular polygon sides=6,minimum size=15] (u) at (10,0) {$ u $};
\node (G) at (7.2,-2.8) {\LARGE $ G_c $};
\node (G) at (17.2,-2.8) {\LARGE $ H_u $};

\ball{20}{0}
\grid{minimum size = 15}{20-2.1}{-2.1}{7}{7}
\node[circle,draw,fill=white!40,minimum size=15] (u) at (20,0) {$ u $};

\end{tikzpicture}
}		
	\caption{Graphical representation of $ G $, $ G_c $ and $ H_u $ ($ n = 4 $). The sub-lattices used to built a bijection between $ G $ and $ G_c $ (cf proof of Lemma \ref{lem}) are represented with distinct shapes and gray levels.}
	\label{fig:eight}
\end{figure}
\end{proof}

\paragraph*{Remarks}
\begin{itemize}
	\item When $ r = 0 $ (uniform shortcut distribution), the bound $ \frac{1}{8} $ is asymptotically tight: the success probability is exactly the ratio between the number of nodes in $ G\setminus\{u\} $ and $ B_u\setminus\{u\} $, which is $ \frac{n^2-1}{4(n-1)(2n-1)} \underset{n\to +\infty}{\longrightarrow} \frac{1}{8} $. On the other hand, as $ r $ grows, the probability mass gets more and more concentrated around $ u $ (hence in $ G $), so we should expect better performance (cf Section \ref{sec:performance-analysis}).
	\item The dynamic rejection sampling approach can be used in other variants of Kleinberg's model, like for other dimensions or when the number of shortcuts per node is a random variable (like in \cite{fraigniaud14greedy}). The only requirement is the existence of some \emph{root} distribution (like the distribution over $ B_u $ here) that can be carved to match any of the possible distributions with a simple acceptance test.
	\item Only the nodes from $ H_u $ may belong to $ G $, so nodes from $ B_u\setminus G $ are always sampled for nothing. For example, in Figure \ref{fig:toyrejection}, this represents 36 nodes over 84. By drawing shortcuts in $ H_u\setminus\{u\} $ instead of $ B_u\setminus\{u\} $, we could increase the success rate lower bound to $ 1/4 $. However, this would make the algorithm implementation more complex (the number of nodes at distance $ i $ is not always $ 4i $), which is in practice not worth the factor 2 improvement of the lower bound. This may not be true for a higher dimension $ \beta $. Adapting the proof from Lemma \ref{lem}, we observe that using a ball of radius $ \beta(n-1) $ will lead to a bound $ \beta ! (2\beta)^{-\beta} $, while a grid of side $ 2n - 1 $ will lead to $2^{-\beta}$. The two bounds are asymptotically tight for $ r = 0 $. In that case the grid approach is $ \frac{\beta ^ \beta}{\beta !} $ more efficient than the ball approach (this  grows exponentially with $ \beta $).
\end{itemize}

\subsection{Performance Evaluation}\label{sec:performance-analysis}

\begin{figure}
\centering

\pgfplotstableread{edt_perf_r_1_with_10000_runs.txt}\ffu
\pgfplotstableread{edt_perf_r_2_with_10000_runs.txt}\ffd
\pgfplotstableread{edt_perf_r_2.5_with_10000_runs.txt}\fft

\begin{tikzpicture}
\begin{loglogaxis}[
tikzDefaults,
width=0.65\textwidth,
xlabel=$ n $,
ylabel= {Running time (seconds)},
legend pos=north west,	
]

\addplot table [x=n, y = time]{\fft};
\addlegendentry{$ r = 2.5 $}

\addplot table [x=n, y = time]{\ffu};
\addlegendentry{$ r = 1 $}

\addplot table [x=n, y = time]{\ffd};
\addlegendentry{$ r = 2 $}

\end{loglogaxis}
\end{tikzpicture}
\caption{Computing $ e_r(n) $ using $ R = 10,000 $ runs}
\label{fig:computation}
\end{figure}
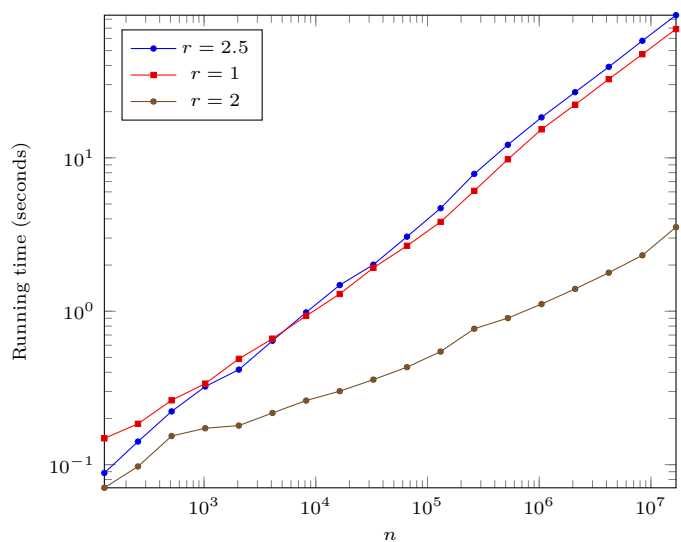

We implemented our algorithm in Julia 0.4.5. A working code example is provided in Appendix \ref{sec:code}. Simulations were executed on a low end device (a Dell tablet with 4 Gb RAM and Intel M-5Y10 processor), which was largely sufficient for getting fast and accurate results on very large grids, thanks to the dynamic rejection sampling.

Unless said otherwise, $ e_r(n) $ is obtained by averaging $ R = 10,000 $ runs. For $ n $, we mainly consider powers of two ranging from $ 2^7 $ (about 16,000 nodes) to $ 2^{24} $ (about 280 trillions nodes). We focus on $ r\in [0, 3] $. 

Regarding the choice of the bulk size $ k $ (we assume here the use of bulk inverse transform sampling), the average number of shortcuts to draw over the $ R $ runs is $ Re_r(n) $. This leads to an average total cost for drawing shortcuts in $ O(\lceil\frac{Re_r(n)}{k} \rceil(n+k\log(n))) $. $ k $ can be optimized if $ e_r(n) $ is known, but this requires bootstrapping the estimation of $ e_r(n) $. Yet, we can remark that for $ n $ fixed and $ R $ large enough, we should choose $ k $ of the same order of magnitude than $ n $, which ensures an average cost per shortcut in $ O(\log(n)) $. This is the choice we made in our code, which gives a complexity in $ O(\max(Re_r(n), n)\log(n)) $. This is not efficient for $ n\gg R e_r(n) $ (the bulk is then over-sized, so lot of unused shortcuts are drawn), but this seldom happens with our settings.

Figure \ref{fig:computation} presents the time needed to compute $ e_r(n) $ as a function of $ n $ for $ r \in \{1, 2, \frac 5 2 \}$.
We observe running times ranging from seconds to a few minutes. To compare with, in \cite{athan10high}, which is to the best of our knowledge the only work disclosing computation times, one single run takes about 4 seconds for $ n = 400 $. With a similar time budget, our implementation averages 10,000 runs for $ n=2^{17}\approx 130,000 $ ($ r = 1 $ or $ r = 2.5 $), or up to $ n = 2^{24} \approx 16,000,000 $ for the optimal exponent $ r = 2 $. In other words, we are several orders of magnitude faster.

\begin{figure}
\centering
\pgfplotstableread{delivery_overhead_n_16384_with_10000_runs.txt}\fn
\begin{tikzpicture}
\begin{axis}[
tikzDefaults,
width=0.65\textwidth,
xlabel=$ r $,
ylabel= {Success rate},
ymax = 1,
ymin = 0,
xmin = 0,
]
\addplot table [x=r, y expr = 1 / \thisrow{overhead}]{\fn};
\end{axis}
\end{tikzpicture}
\caption{Frequency of shortcuts drawn in $ B_u $ that belongs to $ G $ observed during a computation of $ e_r(n) $ ($ n = 2^{14} $)}
\label{fig:rejection}
\end{figure}
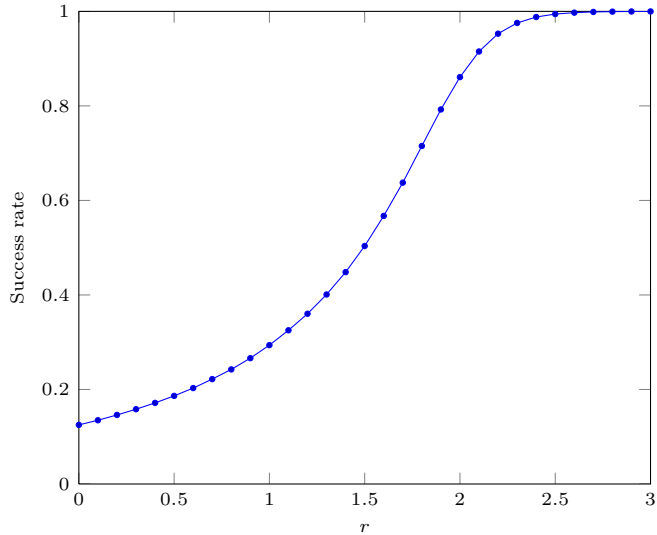

We also looked at the cost of the dynamic rejection sampling approach in terms of shortcuts drawn outside of $ G $. Figure \ref{fig:rejection} shows the success frequency of the sampling as a function of $ r $ for $ n = 2^{14} $ (the actual value of $ n $ has no significant impact as long as it is large enough). We verify Lemma \ref{lem}: the success rate is always at least $ 1/8 = 0.125 $, which is a tight bound for $ r = 0 $. For $ r = 1 $, the rate is about $ 0.29 $, and it climbs to $ 0.86 $ for $ r = 2 $. Failed shortcuts become negligible (rate greater than $ 0.99 $) for $ r\geq 2.5 $. Overall, Figure \ref{fig:rejection} shows that the cost of drawing some shortcuts outside $ G $ is a small compared to the benefits offered by drawing shortcuts from a unique distribution. 

\paragraph*{Remark} In terms of success rate, Figure \ref{fig:rejection} shows that $ r = 2.5 $ is much more efficient than $ r = 1 $, but running times tend to be slightly longer for $ r = 2.5 $  (cf Figure \ref{fig:computation}). The reason is that $ e_1(n) $ is lower than $ e_{2.5}(n) $, which overcompensates the success rate difference.

\section{Applications}\label{sec:applications}

Given the tremendous amount of strong theoretical results on small-world routing, one can question the interest of proposing a simulator (even a fast one!).

In this Section, we prove the interest of numerical evaluation through three (almost) independent small studies.

\subsection{Efficient enough exponents}\label{sec:efficient-enough-exponents}

In \cite{Kleinberg00navigation}, Kleinberg gives an estimation of $ e_r(20,000) $ using a torus approximation (cf Section \ref{sec:torus-approximation}). A few years later, he discussed in more details the results, observing that~\cite{kleinberg10networks}:
\begin{itemize}
	\item $ e_r(n) $ stays quite similar for $ r\in [1.5, 2]$;
	\item the best value of $ r $ is actually slightly lower than $ 2 $.
\end{itemize}

We believe that these observations are very important as they show that routing in Kleinberg's grid is more robust than predicted by theory: it is efficient as long as $ r $ is \emph{close enough} to 2.

Using our simulator, we can perform the same experiment. As shown by Figure \ref{fig:n20000}, the results are quite similar to the ones observed in \cite{Kleinberg00navigation}.

\begin{figure}
	\centering
	\pgfplotstableread{delivery_overhead_n_20000_with_10000_runs.txt}\fn
	\begin{tikzpicture}
	\begin{semilogyaxis}[
	tikzDefaults,
	width=0.65\textwidth,
	xlabel=$ r $,
	ylabel= {$ e_r(20,000) $},
	ytick = {10, 100, 1000},
	ymin = 100
	]
	\addplot table [x=r, y = delivery]{\fn};
	\end{semilogyaxis}
	\end{tikzpicture}
	\caption{Expected delivery time for $ n = 20,000 $}
	\label{fig:n20000}
\end{figure}
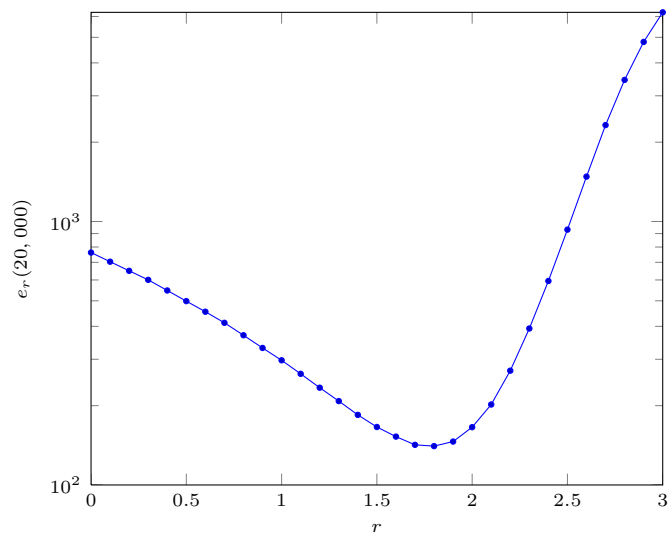

Yet, there is a small but \emph{essential} difference between the two experiments: Kleinberg approximated his grid by a torus while we stay true to the original model. Why do we use the word \emph{essential}? Both shapes have the same asymptotic behavior (proofs in \cite{Kleinberg00navigation} are straightforward to adapt), so why should we care? It seems to us that practical robustness is an essential feature if one wants to accept Kleinberg's grid as a reasonable model for routing in social networks. To the best of our knowledge, no theoretical work gives quantitative results on this robustness, so we need to rely on numerical evaluation. But when Kleinberg uses a torus approximation, we can not rule out that the observed robustness is a by-side effect of the toroidal topology. The observation of a similar phenomenon on a flat grid discards this hypothesis. In fact, it suggests (without proving) that the robustness with respect to the exponent for grids of finite size may be a general phenomenon.

We propose now to investigate this robustness in deeper details. We have evaluated the following values, which outline, for a given $ n $, the values of $ r $ that can be considered reasonable for performing greedy routing:
\begin{itemize}
	\item The value of $ r $ that minimizes $ e_r(n) $, denoted $ r_{\opt}$;
	\item The smallest value of $ r $ such that $ e_r(n) \leq e_2(n) $, denoted $ r_{\min}(e_2(n)) $;
	\item The smallest and largest values of $ r $ such that $ e_r(n) \leq 2 e_2(n) $, denoted $ r_{\min}(2e_2(n)) $ and $ r_{\max}(2e_2(n)) $ respectively.
\end{itemize}

The results are displayed in Figure \ref{fig:interval}. All values but $ r_{\opt} $ are computed by bisection. For $ r_{\opt} $, we use a Golden section search \cite{golden}. Finding a minimum requires more accuracy, so the search of $ r_{\opt} $ is set to use $R =  1,000,000 $ runs per estimation. Luckily, as the computation operates by design through near-optimal values, we can increase the accuracy with reasonable running times.

\begin{figure}
\centering
\pgfplotstableread{kleinberg_stats_with_10000_runs.txt}\fn
\pgfplotstableread{kleinberg_min_r_with_1000000_runs.txt}\frmin
\begin{tikzpicture}
\begin{semilogxaxis}[
tikzDefaults,
width=0.65\textwidth,
xlabel=$ n $,
ylabel= $ r $,
legend pos=south east,	
]
\addplot table [x=n, y = rsup2d2]{\fn};
\addlegendentry{$ r_{\max}(2e_2(n)) $}
\addplot table [x=n, y = rmin]{\frmin};
\addlegendentry{$ r_{\opt}(n) $}
\addplot table [x=n, y = rinfd2]{\fn};
\addlegendentry{$ r_{\min}(e_2(n)) $}
\addplot table [x=n, y = rinf2d2]{\fn};
\addlegendentry{$ r_{\min}(2e_2(n)) $}
\addplot+ [style=dashed, no markers] coordinates
{(128,2) 	(16777216,2)};
\end{semilogxaxis}
\end{tikzpicture}
\caption{Reasonable values of $ r $}
\label{fig:interval}
\end{figure}
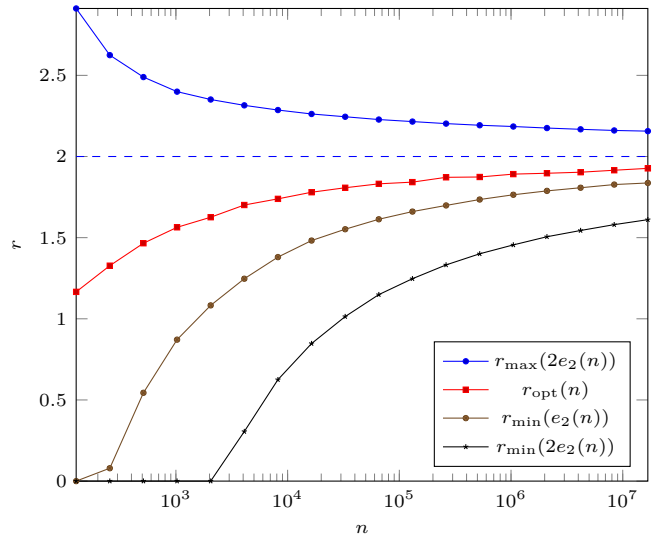

Besides confirming that $ r = 2 $ is asymptotically the optimal value, Figure \ref{fig:interval} shows that the range of reasonable values for finite grids is quite comfortable. For example, considering the range of values where $ e_r(n) $ is less than twice $ e_2(n) $, we observe that:
\begin{itemize}
	\item For $ n\leq 2^{11} $ (less than four million nodes), any $ r $ between 0 and 2.35 works;
	\item For $ n \leq 2^{14} $ (less than 270 million nodes), the range is between 0.85 and 2.26. 
	\item Even for $ n = 2^{24} $ (about 280 trillions nodes), all values of $ r $ between 1.58 and 2.16 can still be considered \emph{efficient enough}.
\end{itemize}

\subsection{Asymptotic Behavior}\label{sec:asymptotic-behavior}

\begin{figure*}[!ht]
\centering
\begin{minipage}{1.1\textwidth}
\pgfplotstableread{edt_perf_r_1_with_10000_runs.txt}\fn
\begin{tikzpicture}
\begin{loglogaxis}[
tikzDefaults,
width=0.34\textwidth,
xlabel=$ n $,
title = {$ r = 1 $},
legend pos=south east,	
]
\addplot table [x=n, y = delivery]{\fn};
\addlegendentry{$ e_1(n) $}
\addplot[domain = 128:16777216] {2*(x)^(1/2)};
\addlegendentry{$ 2\sqrt{n} $}
\end{loglogaxis}
\end{tikzpicture}
\pgfplotstableread{edt_perf_r_2_with_10000_runs.txt}\fn
\begin{tikzpicture}
\begin{loglogaxis}[
tikzDefaults,
width=0.34\textwidth,
xlabel=$ n $,
legend pos=south east,	
title = {$ r = \smash{2} $},
]
\addplot table [x=n, y = delivery]{\fn};
\addlegendentry{$ e_2(n) $}
\addplot[domain = 128:16777216] {2*(ln(x))^2-20};
\addlegendentry{$ \frac{\log^2(n)-10}{2} $}
\end{loglogaxis}
\end{tikzpicture}
\pgfplotstableread{edt_perf_r_2.5_with_10000_runs.txt}\fn
\begin{tikzpicture}
\begin{loglogaxis}[
tikzDefaults,
width=0.34\textwidth,
xlabel=$ n $,
legend pos=south east,	
title = {$ r = \smash{\frac{5}{2}} $},
]
\addplot table [x=n, y = delivery]{\fn};
\addlegendentry{$ e_{2.5}(n) $}
\addplot[domain = 128:16777216] {7*(x)^(0.5)};
\addlegendentry{$ 7 \sqrt{n}$}
\end{loglogaxis}
\end{tikzpicture}
\end{minipage}
\caption{Expected delivery time for different values of $ r $}
\label{fig:varyingr}
\end{figure*}
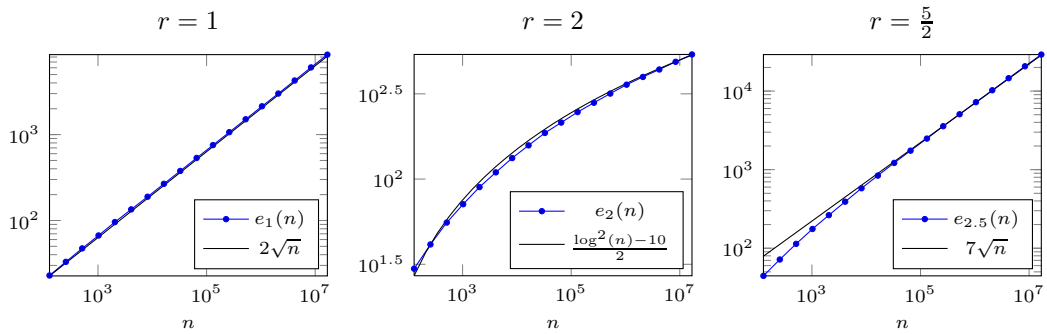

Our simulator can be used to verify the theoretical bounds proposed in  \cite{Kleinberg00thesmall-world}. For example, Figure \ref{fig:varyingr} shows $ e_r(n) $ for $ r $ equal to $ 1 $, 2, and $ \frac{5}{2} $.

As predicted, $ e_r(n) $ seems to behave like $ \log^2(n) $ for $ r = 2 $ and like $ n^\alpha $ for the two other cases. Yet, the exponents found differ from the ones proposed in \cite{Kleinberg00thesmall-world}. For both $ r = 1 $ and $ r = 2.5 $, we observe $ \alpha = \frac{1}{2} $, while the lower bound is $ \frac{1}{3} $. Intrigued by the difference, we want to compute $ \alpha $ as a function of $ r $. 

However, we see in Figure \ref{fig:varyingr} that a $ \log^2(n) $ curve appears to have a positive slope in a logarithmic scale, even for large values of $ n $. This may distort our estimations. To control the possible impact of this distortion, we estimate the exponent at two distinct scales:
\begin{itemize}
	\item  $ n\in [2^{15}, 2^{20}] $, using the estimation $\alpha_1 := \frac{\log_2(e_r(2^{20}))-\log_2(e_r(2^{15}))}{5} $;
	\item $ n\in [2^{20}, 2^{24}] $, using the estimation $ \alpha_2 := \frac{\log_2(e_r(2^{24}))-\log_2(e_r(2^{20}))}{4} $.
\end{itemize}
The Results are displayed in Figure \ref{fig:exponent}. The range of $ r $ was extended to $ [0,4] $ due to the observation of a new transition for $ r = 3 $

As feared, our estimations do not indicate 0 for $ r = 2 $, but the fact that $ \alpha_2 $ is closer to 0 than $ \alpha_1 $ confirms that this is likely caused by the $ \log^2(n) $ factor. Moreover, we observe that $ \alpha_1 $ and $ \alpha_2 $ only differ for $ r \approx 2 $ and $ r\approx 3 $, suggesting that both estimates should be accurate except for these critical values.

\begin{figure}
\centering
\pgfplotstableread{p_exponents_with_10000_runs.txt}\xp

\begin{tikzpicture}
\begin{axis}[
tikzDefaults,
width=0.65\textwidth,
xlabel=$ r $,
ylabel= $ \alpha $,
legend pos=north west,
legend style = {at={(.4,.98)},anchor=north},
ymax = 1,
ymin = 0,
xmin = 0,
clip = false	
]

\addplot table [x=r, y = expo1]{\xp};
\addlegendentry{$ \alpha_1 $}

\addplot table [x=r, y = expo2]{\xp};
\addlegendentry{$ \alpha_2 $}

\addplot[dotted, semithick, domain = 0:2] {(2-x)/(3-x)};
\addlegendentry{$ \frac{2-r}{3-r} $}
\addplot[dashed, domain = 2:3] {(x-2)};

\addlegendentry{$r-2$}
\addplot[domain = 0:2] {(2-x)/3};
\addplot[domain = 2:4] {(x-2)/(x-1)};

\addlegendentry{Lower bound (\cite{Kleinberg00thesmall-world})}
\end{axis}
\end{tikzpicture}
\caption{Estimates of the exponent $ \alpha $}
\label{fig:exponent}
\end{figure}
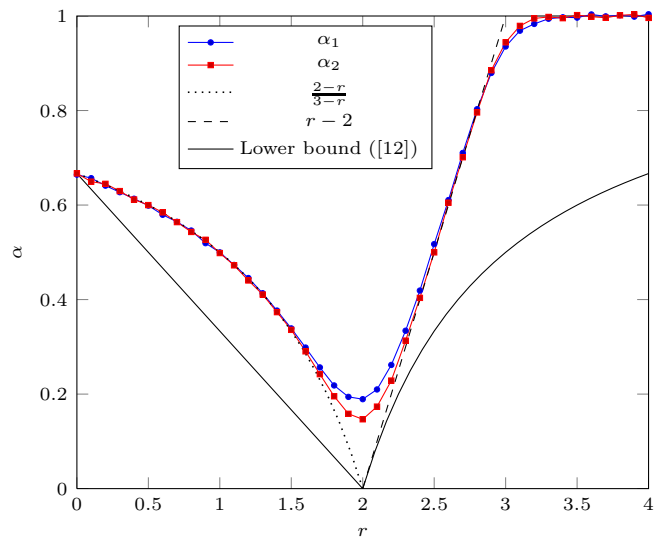

Based on Figures \ref{fig:varyingr} and \ref{fig:exponent}, it seems that the bounds proposed by Kleinberg in \cite{Kleinberg00thesmall-world} are only tight for $ r = 0 $ and $ r = 2 $. Formally proving more accurate bounds is beyond the scope and spirit of this paper, but we can conjecture new bounds hinted by our simulations:
\begin{itemize}
	\item For $ 0\leq r < 2 $, we have $ e_r(n) = \Theta(n^\frac{2-r}{3-r})$;
	\item For $ 2< r < 3 $, we have $ e_r(n) = \Theta(n^{r-2})$;
 	\item For $ r > 3 $, we have $ e_r(n) = \Theta(n)$.
\end{itemize}

This conjectured bounds are consistent with the one proposed in \cite{fraigniaud01efficient}, where it is proved that for the 1-dimensional ring, we have:
\begin{itemize}
	\item For $ 0\leq r < 1 $, $ e_r(n) = \Omega(n^\frac{1-r}{2-r})$;
	\item For $ r = 1 $, $ e_r(n) = \Theta(\log^2(n)) $
	\item For $ 1< r < 2 $, $ e_r(n) = O(n^{r-1})$;
	\item For $ r = 2 $, $ e_r(n) = O(n\frac{\log(\log(n))}{\log(n)}) $
	\item For $ r > 2 $, $ e_r(n) = O(n)$.
\end{itemize}

It is likely that the proofs in \cite{fraigniaud01efficient} can be adapted to the 2-dimensional grid, but some additional work may be necessary to demonstrate the bounds' tightness. Moreover, our estimates may have missed some slower-than-polynomial variations. For example, the logarithms in the bounds for $ r=2 $ in \cite{fraigniaud01efficient} may have a counterpart in the grid for $ r=3 $. This would explain why the estimates are not sharp around that critical value (like for the case $ r=2 $ and the term in $ \log^2(n) $).

\paragraph*{Remark} For $ r\geq 3.5 $, we have observed that $ e_r(n)\approx \frac{2}{3}n $, which is the expected delivery path in absence of shortcuts: for these high values of $ r $, almost all shortcuts link to an immediate neighbor of the current node, which make them useless, and the rare exceptions are so short that they make no noticeable difference.

\subsection{Six degrees of separation}\label{sec:six-degrees-of-separation}

What makes Milgram's experiments \cite{milgram67small,milgram69experimental} so famous is the surprising observation that only a few hops are required to transmit messages in social networks, a phenomenon called \emph{six degrees of separation} in popular culture.

Yet, the values of $ e_r(n) $ observed until now are quite far for the magic number six. For example, in Figure \ref{fig:n20000}, the lowest value is 140. In addition to the asymptotic lack of robustness with respect to the exponent (discussed in Section \ref{sec:efficient-enough-exponents}), this may partially motivate the amount of work accomplished to increase the realism of the model and the performance of the routing algorithm.

In our opinion, a good model should be as simple as possible while being able to accurately predict the phenomenon that needs to be explained.
We propose here to answer a simple question: is the genuine model of greedy routing in $ G(n,r,p,q) $ a good model? Sure, it is quite simple. To discuss its accuracy, we need to tune the four parameters $ n,r,p $ and $ q $ to fit the conditions of Milgram's experiments as honestly as we can.
\begin{description}
	\item[Size] The experiments of Milgram were conducted in the United States, with a population of about 200,000,000 at the late sixties. All inhabitants were not susceptible to participate to the experiments: under-aged, undergraduate or disadvantaged people may be considered as \emph{de facto} excluded from the experiments. Taking that into consideration, the correct $ n $ is probably somewhere in the range $ [5000, 14000] $. We propose to set $ n = 8,500 $, which corresponds to about 72,000,000 potential subjects.
	\item[Exponent] In \cite{kleinberg10networks}, Kleinberg investigates how to relate the $r$-harmonic distribution with real-life observations. He surveys multiple social experiments and discusses the correspondence with the exponent of his model, which gives estimates of $ r $ between 1.75 and 2.2.
	\item[Neighborhood] The default value $ p = q = 1 $ means that there are no more than five ``acquaintances'' per node. This is quite small compared to what is observed in real-life social networks. For example, the famous Dunbar's number, which estimates the number of \emph{active} relationships, is 150 \cite{dunbar}. More recent studies seem to indicate that the average number of acquaintances is larger, ranging from 250 to 1500 (see \cite{pool1978contacts,mccormick10people,wellman:dunbar} and references within). We propose to set $ p $ and $ q $ so that the neighborhood size $ 2p(p+1) + q $ is about 600, the value reported in \cite{mccormick10people}. Regarding the partition between local links ($ p $) and shortcuts ($ q $), we consider three typical scenarios:
	\begin{itemize}
		\item $ p = 1 $, $ q = 600 $ (shortcut scenario: the neighborhood is almost exclusively made of shortcuts, and local links are only here to ensure the termination of greedy routing).
		\item $ p = 10 $, $ q = 380 $ (balanced scenario).
		\item $ p = 15 $, $ q = 120 $ (local scenario, with a value of $ q $ not too far from Dunbar's number).
	\end{itemize}
\end{description}

Having set all parameters, we can evaluate the performance of greedy routing. The results are displayed in Figure \ref{fig:six}. We observe that the expected delivery time roughly stands between five and six for a wide range of exponents.
\begin{itemize}
	\item $ r\in [1.4, 2.3] $ for the shortcut scenario.
	\item $ r\in [1.3, 2.3] $ for the balanced scenario.
	\item $ r\in [1.3, 2] $ for the local scenario.		
\end{itemize}
Except for the local scenario, which leads to slightly higher routing times for $ r>2 $, the six degrees of separation are achieved for all values of $ r $ that are consistent with the observations surveyed in \cite{kleinberg10networks}. This allows to answer our question: 
the augmented grid proposed by Kleinberg is indeed a good model to explain the \emph{six degrees of separation} phenomenon.

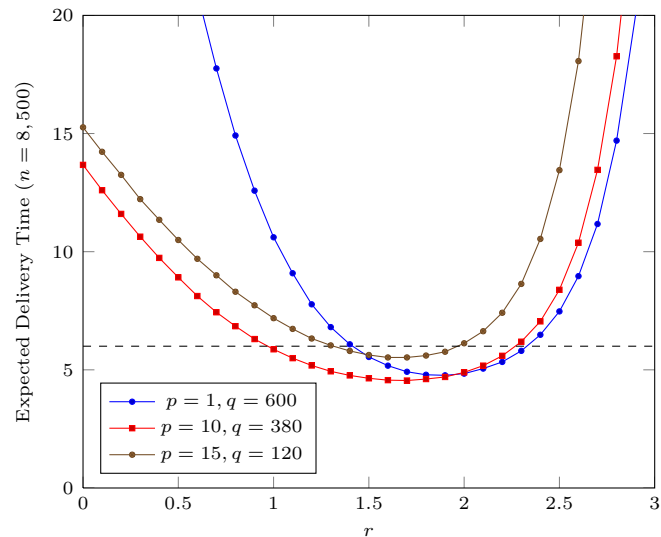
\begin{figure}
\centering
\pgfplotstableread{six_degree2_n_8500_p_1_q_600_10000_runs.txt}\fn
\pgfplotstableread{six_degree2_n_8500_p_10_q_380_10000_runs.txt}\fd
\pgfplotstableread{six_degree2_n_8500_p_15_q_120_10000_runs.txt}\ft
\begin{tikzpicture}
\begin{axis}[
tikzDefaults,
width=0.65\textwidth,
xlabel=$ r $,
ylabel= {Expected Delivery Time ($ n = 8,500 $)},
legend pos=south west,
ymax = 20,
ymin = 0
]
\addplot table [x=r, y = apl]{\fn};
\addlegendentry{$ p = 1, q = 600 $}
\addplot table [x=r, y = apl]{\fd};
\addlegendentry{$ p = 10, q = 380 $}
\addplot table [x=r, y = apl]{\ft};
\addlegendentry{$ p = 15, q = 120 $}
\addplot+ [style=dashed, no markers] coordinates
{(0,6) 	(3,6)};
\end{axis}
\end{tikzpicture}
\caption{Performance of greedy routing for parameters inspired by Milgram's experiments.}
\label{fig:six}
\end{figure}

\newpage

\section{Conclusion}

We proposed an algorithm to evaluate the performance of greedy routing in Kleinberg's grid. Fueled by a dynamic rejection sampling approach, the simulator based on our solution performs several orders of magnitude faster than previous attempts. It allowed us to investigate greedy routing under multiple perspective.
\begin{itemize}
	\item We noted that the performance of greedy routing is less sensitive to the choice of the exponent $ r $ than predicted by the asymptotic behavior, even for very large grids.
	\item We observed that the bounds proposed in \cite{Kleinberg00thesmall-world} are not tight except for $ r=0 $ and $ r = 2 $. We conjectured that the tight bounds are the 2-dimensional equivalent of bounds proposed in \cite{fraigniaud01efficient} for the 1-dimensional ring.
	\item We claimed that the model proposed by Kleinberg in \cite{Kleinberg00thesmall-world,Kleinberg00navigation} is a good model for the \emph{six degrees of separation}, in the sense that it is very simple \emph{and} accurate.
\end{itemize}

Our simulator is intended  as a tool to suggest and evaluate theoretical results, and possibly to build another bridge between theoretical and empirical study of social systems. We hope that it will be useful for researchers from both communities. 
In a future work, we plan to make our simulator more generic so it can handle other types of graph and augmenting schemes.

\newpage

\appendix

\section{Code for Expected Delivery Time (tested on Julia Version 0.4.5)}\label{sec:code}

\begin{lstlisting}
using StatsBase # For using Julia built-in sample function

function shortcuts_bulk(n, probas, bulk_size)
	radii = sample(1:(2*n-2), probas, bulk_size)
	shortcuts = Tuple{Int, Int}[]
	for i = 1:bulk_size
		radius = radii[i]
		angle = floor(4*radius*rand())-2*radius
		push!(shortcuts, ((radius-abs(angle)),
			(sign(angle)*(radius-abs(radius-abs(angle))))))
	end
	return shortcuts
end

# Estimates the Expected Delivery Time for G(n, r, p, q) over R runs
function edt(n, r, p, q, R)
	bulk_size = n
	probas = weights(1./(1:(2*n-2)).^(r-1))
	shortcuts = shortcuts_bulk(n, probas, bulk_size)
	steps = 0
	for i = 1:R
		# s: start/current node; a: target node; d: distance to target
		s_x, s_y, a_x, a_y = tuple(rand(0:(n-1), 4)...) 
		d = abs(s_x - a_x) + abs(s_y - a_y)
		while d>0
			sh_x, sh_y = -1, -1 # sh will be best shortcut node
			d_s = 2*n # d_s will be distance from sh to a
			for j = 1:q # Draw q shortcuts
				ch_x, ch_y = -1, -1 # ch will be current shortcut
				c_s = 2*n # c_s will be distance from ch to a
				# Dynamic rejection sampling
				while (ch_x < 0 || ch_x >= n || ch_y < 0 || ch_y >= n)
					r_x,r_y = pop!(shortcuts)
					ch_x, ch_y = s_x + r_x,  s_y + r_y
					if isempty(shortcuts)
						shortcuts = shortcuts_bulk(n, probas, bulk_size)
					end
				end
				c_s = abs(a_x - ch_x) + abs(a_y - ch_y)
				if c_s < d_s # maintain best shortcut found
					d_s = c_s
					sh_x, sh_y = ch_x, ch_y
				end
			end
			if d_s < d-p # Follow shortcut if efficient
				s_x, s_y  = sh_x, sh_y
				d = d_s
			else # Follow local links
				d = d - p
				delta_x = min(p,abs(a_x - s_x))
				delta_y = p - delta_x
				s_x += delta_x*sign(a_x - s_x)
				s_y += delta_y*sign(a_y - s_y)
			end
			steps += 1
		end
	end
	steps /= R
	return steps
end
\end{lstlisting}


\begin{thebibliography}{10}

\bibitem{athan10high}
Stavros Athanassopoulos, Christos Kaklamanis, Ilias Laftsidis, and Evi
  Papaioannou.
\newblock An experimental study of greedy routing algorithms.
\newblock In {\em High Performance Computing and Simulation (HPCS), 2010
  International Conference on}, pages 150--156, June 2010.

\bibitem{fraigniaud01efficient}
Lali Barri\`{e}re, Pierre Fraigniaud, Evangelos Kranakis, and Danny Krizanc.
\newblock Efficient routing in networks with long range contacts.
\newblock In {\em Proceedings of the 15th International Conference on
  Distributed Computing}, DISC '01, pages 270--284, London, UK, 2001.

\bibitem{pool1978contacts}
Ithiel de~Sola~Pool and Manfred Kochen.
\newblock Contacts and influence.
\newblock {\em Social Networks}, 1:5--51, 1978.

\bibitem{dunbar}
R.~I.~M. Dunbar.
\newblock {Neocortex size as a constraint on group size in primates}.
\newblock {\em Journal of Human Evolution}, 22(6):469--493, June 1992.

\bibitem{kleinberg10networks}
David Easley and Jon Kleinberg.
\newblock The small-world phenomenon.
\newblock In {\em Networks, Crowds, and Markets: Reasoning About a Highly
  Connected World}, chapter~20, pages 611--644. Cambridge University Press,
  2010.

\bibitem{fraigniaud09universal}
Pierre Fraigniaud, Cyril Gavoille, Adrian Kosowski, Emmanuelle Lebhar, and Zvi
  Lotker.
\newblock {Universal Augmentation Schemes for Network Navigability: Overcoming
  the $\sqrt(n)$-Barrier}.
\newblock {\em {Theoretical Computer Science}}, 410(21-23):1970--1981, 2009.

\bibitem{fraigniaud2006eclectism}
Pierre Fraigniaud, Cyril Gavoille, and Christophe Paul.
\newblock Eclecticism shrinks even small worlds.
\newblock {\em Distributed Computing}, 18(4):279--291, 2006.

\bibitem{stoc2010_searchability}
Pierre Fraigniaud and George Giakkoupis.
\newblock On the searchability of small-world networks with arbitrary
  underlying structure.
\newblock In {\em Proceedings of the 42nd {ACM} Symposium on Theory of
  Computing (STOC)}, pages 389--398, June~6--8 2010.

\bibitem{fraigniaud14greedy}
Pierre Fraigniaud and George Giakkoupis.
\newblock Greedy routing in small-world networks with power-law degrees.
\newblock {\em Distributed Computing}, 27(4):231--253, 2014.

\bibitem{karinthy29chains}
Frigyes Karinthy.
\newblock Láncszemek, 1929.

\bibitem{Kleinberg00navigation}
Jon Kleinberg.
\newblock Navigation in a small world.
\newblock {\em Nature}, August 2000.

\bibitem{Kleinberg00thesmall-world}
Jon Kleinberg.
\newblock The small-world phenomenon: An algorithmic perspective.
\newblock In {\em in Proceedings of the 32nd ACM Symposium on Theory of
  Computing}, pages 163--170, 2000.

\bibitem{Liben-Nowell16082005}
David Liben-Nowell, Jasmine Novak, Ravi Kumar, Prabhakar Raghavan, and Andrew
  Tomkins.
\newblock Geographic routing in social networks.
\newblock {\em Proceedings of the National Academy of Sciences of the United
  States of America}, 102(33):11623--11628, 2005.

\bibitem{mankun04eighbor}
Gurmeet~Singh Manku, Moni Naor, and Udi Wieder.
\newblock Know thy neighbor's neighbor: The power of lookahead in randomized
  {P2P} networks.
\newblock In {\em Proceedings of the Thirty-sixth Annual ACM Symposium on
  Theory of Computing (STOC)}, pages 54--63. ACM, 2004.

\bibitem{martel04analyzing}
Chip Martel and Van Nguyen.
\newblock Analyzing kleinberg's (and other) small-world models.
\newblock In {\em Proceedings of the Twenty-third Annual ACM Symposium on
  Principles of Distributed Computing}, PODC '04, pages 179--188, New York, NY,
  USA, 2004. ACM.

\bibitem{mccormick10people}
Tyler~H. McCormick, Matthew~J. Salganik, and Tian Zheng.
\newblock How many people do you know?: Efficiently estimating personal network
  size.
\newblock {\em Journal of the American Statistical Association},
  105(489):59--70, 2010.

\bibitem{milgram67small}
Stanley Milgram.
\newblock The small world problem.
\newblock {\em Psychology Today}, 67(1):61--67, 1967.

\bibitem{Mitzenmacher:2005:PCR:1076315}
Michael Mitzenmacher and Eli Upfal.
\newblock {\em Probability and Computing: Randomized Algorithms and
  Probabilistic Analysis}.
\newblock Cambridge University Press, New York, NY, USA, 2005.

\bibitem{golden}
William~H. Press, Saul~A. Teukolsky, William~T. Vetterling, and Brian~P.
  Flannery.
\newblock Minimization or maximization of functions.
\newblock In {\em Numerical Recipes 3rd Edition: The Art of Scientific
  Computing}, chapter~10. Cambridge University Press, New York, NY, USA, 3
  edition, 2007.

\bibitem{milgram69experimental}
Jeffrey Travers and Stanley Milgram.
\newblock An experimental study of the small world problem.
\newblock {\em Sociometry}, 32:425--443, 1969.

\bibitem{rejection}
John von Neumann.
\newblock Various techniques used in connection with random digits.
\newblock {\em Nat. Bureau Standards}, 12:36–38, 1951.

\bibitem{wellman:dunbar}
Barry Wellman.
\newblock Is {Dunbar's} number up?
\newblock {\em British Journal of Psychology}, 2011.

\end{thebibliography}
\end{document}